\documentclass[conference]{IEEEtran}
\IEEEoverridecommandlockouts

\usepackage{makecell}
\usepackage{cite}
\usepackage{amsmath,amssymb,amsfonts}
\usepackage{multirow} 
\usepackage{algorithmic}
\usepackage{graphicx}
\usepackage{textcomp}
\usepackage{xcolor}

\usepackage{tabularx}
\usepackage{array}
\usepackage{balance}
\usepackage[utf8]{inputenc}
\usepackage{amsthm}
\usepackage{bm}
\usepackage{booktabs}
\usepackage{arydshln}
\usepackage{placeins}
\definecolor{mygreen}{RGB}{0,140,110}
\usepackage{float}
\floatstyle{ruled}
\newfloat{problem}{thp}{lop}
\floatname{problem}{Problem}
\newfloat{model}{thp}{lop}
\floatname{model}{Model}

\definecolor{niceblue}{rgb}{0, 0.5, 1.0}
\usepackage[
    colorlinks=true,
    linkcolor=niceblue,
    citecolor=niceblue,
    urlcolor=niceblue
]{hyperref}

\newtheorem{remark}{Remark}
\newtheorem{lemma}{Lemma}

\def\BibTeX{{\rm B\kern-.05em{\sc i\kern-.025em b}\kern-.08em
    T\kern-.1667em\lower.7ex\hbox{E}\kern-.125emX}}

\title{Activate the Dual Cones: A Tight Reformulation of Conic ACOPF Constraints}

\author{%
\IEEEauthorblockN{Saba Rafiei and Samuel Chevalier}
\IEEEauthorblockA{\textit{Department of Electrical and Biomedical Engineering}\\
\textit{University of Vermont}, Burlington, VT, USA\\
Email: \{srafiei, schevali\}@uvm.edu}
}

\begin{document}
\maketitle

\begin{abstract}
By exploiting the observed tightness of dual rotated second-order cone (RSOC) constraints, this paper transforms the dual of a conic ACOPF relaxation into an equivalent, non-conic problem where dual constraints are implicitly enforced through eliminated dual RSOC variables.
To accomplish this, we apply the RSOC-based Jabr relaxation of ACOPF, pose its dual, and then show that all dual RSOC constraints must be tight (i.e., active) at optimality. We then construct a reduced dual maximization problem with only non‐negativity constraints, avoiding the explicit RSOC inequality constraints. Numerical experiments confirm that the tight formulation recovers the same dual objective values as a mature conic solver (e.g., MOSEK via PowerModels) on various PGLib benchmark test systems (ranging from 3 to 1354 buses). The proposed formulation offers numerical advantages over its conic counterpart, and it allows us to define a bounding function that provides a guaranteed lower bound on system cost.
While this paper focuses on demonstrating the correctness and validity of the proposed structural simplification, it lays the groundwork for future GPU-accelerated first-order optimization methods that can exploit the simple bound-constrained structure of the proposed formulation.
\end{abstract}


\section{Introduction}

The AC Optimal Power Flow (ACOPF) problem, which enforces the balance of real and reactive power subject to Kirchhoff’s laws and network constraints, is nonconvex due to the quadratic power flow equations. To address this nonconvexity, researchers have turned to convex relaxations of the full ACOPF formulation.
 Convex relaxation techniques have emerged as promising tools to obtain lower bounds or even exact solutions, under specific structural conditions~\cite{1,2}. 
Among these, Second-Order Cone Programming (SOCP) relaxations have become particularly popular. Jabr developed a reformulation of the power flow equations for radial (tree-structured) distribution networks using rotated second-order cones\cite{3}. This method replaced the nonconvex voltage coupling between buses with a convex conic constraint.

Many enhanced relaxations have made it feasible to obtain provably global or near-global solutions for OPF on large-scale networks (hundreds or thousands of buses) with much less computational effort than global nonconvex solvers~\cite{2,4}. Still, limitations remain. Even the best conic relaxations can have a residual optimality gap in some cases, and each new constraint added for tightness increases the solve time~\cite{5}. Prior work has explored Lagrangian-dual decompositions, where network-coupling constraints are relaxed, and the resulting dual is maximized via bundle or subgradient updates to obtain SOC-level lower bounds~\cite{6}. This method, however, still requires solving SOC subproblems iteratively, and their reported dual bounds can remain loose. This motivates algorithmic advances to solve conic formulations more efficiently.

The traditional approach for solving SOCPs and SDPs is interior-point methods (IPMs). In ACOPF relaxations, scalability challenges are particularly acute for SDP formulations. Even when exploiting network sparsity, IPM-based solutions of large-scale SDPs can become prohibitively expensive for real-world networks with thousands of buses~\cite{7}. This has motivated the use of computationally less demanding convex relaxations such as SOCP and LP-based approximations~\cite{7}, as well as tight surrogates that retain much of the bound quality of SOCP~\cite{8}. In the same spirit, researchers have developed first-order and dual methods that better exploit problem structure and parallelism. A widely studied approach in this class is the Alternating Direction Method of Multipliers (ADMM), which aims to reduce the per-iteration computational cost significantly~\cite{9,10}.





With GPU-accelerated optimization tools becoming increasingly available, first-order methods have gained popularity. Google, Gurobi, and NVIDIA have all recently released Primal-Dual Hybrid Gradient (PDHG) based solvers~\cite{applegate2025pdlp,ccorduk2025gpu,rothberg2025concurrent}. Applying projected gradient methods to conic relaxations of ACOPF is an emerging research direction~\cite{QIU2024110661,tanneau2024dual}.
Compared to general-purpose solvers like the Splitting Conic Solver (SCS), which applies an operator-splitting approach to a homogeneous self-dual embedding of the primal-dual pair~\cite{11}, projected and proximal gradient methods avoid matrix factorizations by relying only on simple gradient steps and projection/proximity operations~\cite{12}. This makes each iteration lightweight and customizable. Each iteration involves only local operations (no matrix factorizations), and updates can often be executed in parallel across buses or lines, making them highly suitable for GPU acceleration and large-scale power systems.
Recent results illustrate that first-order methods with parallelization can overcome the scalability limitations of classical solvers\cite{15,16}. 

To apply gradient-based solution methods to the ACOPF relaxation, efficient conic projection procedures are needed. Differentiable conic projections were proposed by Tanneau in~\cite{tanneau2024dual} for dual Lagrangian learning, and optimal RSOC projections were recently proposed by Bienstock in~\cite{Bienstock_proj} for cutting plane discovery. A key complication is that primal conic constraints are generally \textit{coupled}, often sharing voltage magnitude variables. Conic constraints in the dual domain, however, are decoupled, so optimal projection can occur in parallel. While attempting to derive an efficient dual cone projection procedure for gradient-based methods, we were able to prove that projecting onto the dual cone was unnecessary: at optimality, all dual conic constraints were already tight. This observation has inspired this paper.

In this paper, we do not apply first-order methods to solve ACOPF problems; instead, we lay the groundwork for future such methods by establishing a key result: conic projections in the dual domain are unnecessary. Toward this goal, our paper's contributions are as follows:

    
  \begin{itemize}
    \item First, in the dual of the Jabr-relaxed ACOPF, we prove that the dual RSOC constraints arising from (i) the voltage-product relaxation, (ii) the line apparent-power limits, and (iii) the generator quadratic-cost epigraph are all tight at optimality, eliminating the need for explicit RSOC inequalities in the relaxed ACOPF dual.
    \item Second, exploiting this tightness, we construct an equivalent non-conic reformulation of the ACOPF dual in which the RSOC constraints are implicitly enforced via variable elimination; from this, we derive a certified lower bound on the ACOPF generation cost.
    \item Finally, we demonstrate that our proposed formulation performs reliably as the size of the system increases, and it can outperform the conic formulations solved via commercial solvers in large test cases with negligible loss in solution quality.

\end{itemize}




Beyond the theoretical simplification, our proposed formulation is often
substantially easier for general-purpose interior-point solvers to
handle than the explicit conic dual, particularly on the largest
instances. On the 1354-bus case, Knitro solves our proposed model 45$\times$
faster than the conic model at identical accuracy. We do not claim a runtime
advantage over dedicated native conic solvers such as MOSEK; our
goal is to demonstrate the structural result and the certified bound.
These innovations
will enable future first-order methods applied to conic relaxations of ACOPF.

\section{Methodology}
\label{sec:Methodology}

This section formulates ACOPF and the convex relaxation used by our method. We adopt the standard branch-flow model in polar coordinates and follow with the second-order cone programming (SOCP) relaxation proposed by Jabr~\cite{3}.

\subsection{ACOPF Formulation}

Let $\mathcal{N}$ denote the set of buses and $\mathcal{L} \subseteq \mathcal{N} \times \mathcal{N}$ the set of transmission lines. For each bus $i \in \mathcal{N}$, let $V_i \in \mathbb{C}$ denote the complex voltage, and define $v_i\triangleq|V_i|$ and $\theta_i$ as its magnitude and phase angle, respectively. Each generator is associated with complex power injection $S_k^g = P_k^g + jQ_k^g$. The load at bus $i$ is denoted by $S_i^d = P_i^d + jQ_i^d$.
ACOPF seeks to minimize total generation cost subject to nonlinear power flow constraints and engineering limits, as stated in Model \ref{model:nonconvex}.
\begin{model}[h]
\caption{\hspace{-0.1cm}\textbf{:} Nonconvex ACOPF}
\label{model:nonconvex}
\vspace{-0.25cm}
\begin{subequations} \label{eq:acopf}
\begin{align}\min_{P^g,Q^g,V}\quad & \sum_{k\in\mathcal{G}}C_{k}(P_{k}^{g})\label{eq:cost}\\
\text{s.t.}\quad & S_{ij}\;=\;Y_{ij}^{*}\,V_{i}V_{i}^{*}\;-\;Y_{ij}^{*}\,V_{i}V_{j}^{*},&&\forall(i,j)\in\mathcal{L},\label{eq:Sij_def}\\
 & \sum_{k\in\mathcal{G}_{i}}S_{k}^{g}-S_{i}^{d}\;=\;\sum_{j:(i,j)\in\mathcal{L}}S_{ij},&&\forall i\in\mathcal{N},\label{eq:balance}\\
 & |V_{i}|^{\min}\le|V_{i}|\le|V_{i}|^{\max},&&\forall i\in\mathcal{N},\label{eq:volt}\\
 & P_{k}^{\min}\le P_{k}^{g}\le P_{k}^{\max},&&\forall k\in\mathcal{G},\\
 & Q_{k}^{\min}\le Q_{k}^{g}\le Q_{k}^{\max},&&\forall k\in\mathcal{G},\label{eq:gen}\\
 & |S_{ij}|\le S_{ij}^{\max},&&\forall(i,j)\in\mathcal{L}.\label{eq:thermal}
\end{align}
\end{subequations} 
\vspace{-0.5cm}
\end{model}
Here, $C_{k}(P_{k}^{g})=c^q_{k}\left(P_{k}^{g}\right)^{2}+c^l_{k}P_{k}^{g}$ is the convex quadratic cost function of generator $k$. Constraints~\eqref{eq:Sij_def} and \eqref{eq:balance} enforce the physical power flow and nodal power balance, respectively. Constraint~\eqref{eq:volt} imposes voltage magnitude limits at each bus, constraint~\eqref{eq:gen} ensures generator outputs remain within their capability limits, and constraint~\eqref{eq:thermal} enforces thermal limits on transmission lines.

\subsection{Implementation of Jabr's SOCP Relaxation}

To convexify the ACOPF problem, we adopt the second-order cone programming (SOCP) relaxation introduced by Jabr~\cite{3}, implemented in the real domain using lifted variables. This relaxation eliminates trigonometric terms by working directly in rectangular coordinates and lifting bilinear voltage products. For each transmission line $(i,j) \in \mathcal{L}$, we define
\begin{subequations}
\begin{align}
    c_{ij} &\triangleq  v_i v_j \cos(\theta_i - \theta_j),\label{eq: c_ij} \\
    s_{ij} &\triangleq v_i v_j \sin(\theta_i - \theta_j).
\end{align}
\end{subequations}
For each bus $i \in \mathcal{N}$, we define the squared voltage magnitude:
\begin{align}
    w_i \triangleq |v_i|^2.
\end{align}
We derive the active and reactive power flows in both directions based on the lifted voltages and the physical line admittance parameters:
\begin{subequations}\label{eq: PQij}
\begin{align}
 P_{ij} & =G_{ii}w_{i}+G_{ij}c_{ij}+B_{ij}s_{ij},\\
P_{ji} & =G_{jj}w_{j}+G_{ji}c_{ij}-B_{ij}s_{ij},\\
Q_{ij} & =-B_{ij}w_{i}-G_{ij}s_{ij}+B_{ij}c_{ij},\\
Q_{ji} & =-B_{jj}w_{j}+G_{ji}s_{ij}+B_{ji}c_{ij},
\end{align}
\end{subequations}
which are linear in the lifted variables. The key nonconvex constraints in this lifted space are given by
\begin{align}\label{eq: c2s2}
w_{i}w_{j} &= s_{ij}^{2}+c_{ij}^{2},\\
\arctan\left({s_{ij}}/{c_{ij}}\right)& =\theta_{i}-\theta_{j}.\label{eq: cycle}
\end{align}
Let $Q_r^n$ denote the $n$-dimensional rotated second-order cone,
\[
Q_r^n
=
\left\{
(x,y,z)\in \mathbb{R}\times \mathbb{R}\times \mathbb{R}^{n-2}
\;:\;
\begin{array}{l}
2xy \ge \|z\|_2^2,\\
x, y \ge 0
\end{array}
\right\}.
\]
In the Jabr convex relaxation, \eqref{eq: c2s2} is relaxed to a rotated second-order cone (RSOC) constraint:
\begin{align}\label{eq:jabr}
w_{i}w_{j}\ge c_{ij}^{2}+s_{ij}^{2}\quad\text{or}\quad\left(\tfrac{1}{2}w_{i},\,w_{j},\,[c_{ij},s_{ij}]\right)\in\mathcal{Q}_{r}^{4}.\end{align}

Meanwhile, the cycle constraint \eqref{eq: cycle} can be replaced with relaxed polynomial constraints~\cite{7}. In this paper, we drop the cycle constraint entirely. 

To fully transform \eqref{eq:acopf} into a conic problem, we also model the quadratic cost function using the scaled RSOC constraint
 $t \ge  \sum_{k\in\mathcal{G}}c^q_{k}\left(P_{k}^{g}\right)^{2},$
such that $\sum_{k\in\mathcal{G}}C_{k}(P_{k}^{g}) = t+\sum_{k\in\mathcal{G}}P^g_k c^{l}_k$.
Finally, we represent line limits as SOC constraints:
\begin{equation}
S_{ij}^{2,\max}\ge P_{ij}^{2}+Q_{ij}^{2},
\qquad
S_{ji}^{2,\max}\ge P_{ji}^{2}+Q_{ji}^{2}.
\end{equation}

\subsection{Canonicalization of Jabr's Relaxation}

In this subsection, we canonicalize the relaxed ACOPF model by collecting all primal variables into a vector $x$ via
\vspace{-4mm}
\begin{align}\label{eq: x_vec}
x={\rm vec}\left(t,P^{g},Q^{g},P^{s},Q^{s},P^{r},Q^{r},s,c,w\right),
\end{align}
where $P^{s}$ is the vector of powers flowing from the source bus ($P_{ij}$), $P^{r}$ is the vector of powers flowing from the receiving bus ($P_{ji}$), $c$ is the vector of lifted voltages from \eqref{eq: c_ij}, etc. Next, we encode the linear power balance constraints \eqref{eq:balance} and the line flow relations \eqref{eq: PQij}
within the linear constraint block in~\eqref{eq:primal-ccqp-b}.
Meanwhile, we capture the inequality constraints associated with all variable bounds via box constraint~\eqref{eq:primal-ccqp-c}. For reasons which will be explained, we also include constraints associated with the flow vectors $P^{s}$, $P^{r}$, $Q^{s}$, $Q^{r}$, such that $\ensuremath{Cx+d\le0}\;\Leftrightarrow\;\underline{x}\le x\le\overline{x}$. Lifted voltages are also appropriately bounded. We construct the objective function vector $m$ by concatenating  linear generator cost terms and a selection value, 1, to capture $t$'s quadratic contribution: 
\begin{align}\label{eq: m}
m={\rm vec}(1,c^{l},{\bf 0}). 
\end{align}
In Model~\ref{model:full_rsoc}, we state the full RSOC relaxation of the OPF problem, where each RSOC constraint is given in canonical form (i.e., $2xy\ge z^Tz$). Each constraint is stated with its corresponding dual variables. In the RSOC dual variable tuples \eqref{eq:rsoc1}-\eqref{eq:rsoc4}, the first two elements in the tuple are both scalars; the last element is an appropriately sized vector. Constraint~\eqref{eq:primal-ccqp-b} includes power balance, line flow definitions, and the definition of the auxiliary variables for the straight RSOC formulation. Constraint~\eqref{eq:primal-ccqp-c} represents the box constraint of the variables.
\begin{model}[h]
\caption{\hspace{-0.1cm}\textbf{:} Primal RSOC Relaxation}
\label{model:full_rsoc}
\vspace{-0.25cm}
\begin{subequations}\label{eq: full_rsoc}
\begin{align}
\min_{x}\quad & m^Tx\\
\text{s.t.}\quad & Ax+b=0 &  & \!\!\!:\lambda,\label{eq:primal-ccqp-b}\\[2pt]
 & Cx+d\le0 &  & \!\!\!:\mu,\label{eq:primal-ccqp-c}\\
 &\!\!\!\forall (i,j) \in {\mathcal L}: \nonumber\\
 & 2\left(\tfrac{1}{2}w_{i}\right)w_{j}\ge c_{ij}^{2}+s_{ij}^{2} &  & \!\!\!:(d_{ij}^{v_{1}},d_{ij}^{v_{2}},d_{ij}^{v}),\label{eq:rsoc1}\\
 & 2\left(\tfrac{1}{2}\right)S_{ij}^{2,{\rm max}}\ge P_{ij}^{2}+Q_{ij}^{2} &  & \!\!\!:(d_{ij}^{s_{1}},d_{ij}^{s_{2}},d_{ij}^{s}),\label{eq:rsoc2}\\
 & 2\left(\tfrac{1}{2}\right)S_{ij}^{2,{\rm max}}\ge P_{ji}^{2}+Q_{ji}^{2} &  & \!\!\!:(d_{ji}^{r_{1}},d_{ji}^{r_{2}},d_{ji}^{r}),\label{eq:rsoc3}\\
 & 2\left(\tfrac{1}{2}\right)t\ge\sum_{k\in\mathcal{G}}\left({(c_{k}^{q})^{0.5}}P_{k}^{g}\right)^{2} &  & \!\!\!:(d^{t_{1}},d^{t_{2}},d^{t}).\label{eq:rsoc4}
\end{align}
\end{subequations}
\vspace{-0.25cm}
\end{model}

The dual variable tuples associated with the RSOC constraints each belong to their respective dual RSOCs:
\begin{subequations}\label{eq:dual_tuples}
\begin{align}
&(d_{ij}^{v_{1}},d_{ij}^{v_{2}},d_{ij}^{v}),\ 
(d_{ij}^{s_{1}},d_{ij}^{s_{2}},d_{ij}^{s}),\ 
(d_{ji}^{r_{1}},d_{ji}^{r_{2}},d_{ji}^{r})
\in\mathcal{Q}_{r}^{4,*}, \\
&(d^{t_{1}},d^{t_{2}},d^{t}) 
\in\mathcal{Q}_{r}^{2+n_{g},*}. 
\end{align}
\end{subequations}
where $n_g=|{\mathcal G}|$. The RSOC is self-dual: $\mathcal{Q}_{r}^{n}=\mathcal{Q}_{r}^{n,*}$~\cite{Boyd_Vandenberghe_2004}.

\subsection{Formulating the Dual}
We construct the Lagrangian by dualizing all constraints. We first collect all dual RSOC variables into vector ${\tilde d}$:
\begin{align}
\tilde{d}={\rm vec}(d^{v_{1}}\!,d^{v_{2}}\!,d^{v}\!,d^{s_{1}}\!,d^{s_{2}}\!,d^{s}\!,d^{r_{1}}\!,d^{r_{2}}\!,d^{r}\!,d^{t_{1}}\!,d^{t_{2}}\!,d^{t}).
\end{align}
We cleanly capture the RSOC dualization by introducing a linear function of the primal variables, $Fx+g$, such that the inner product of the dual RSOC variable vector ${\tilde d}$ and $Fx+g$ is a valid dualization (see Appendix \ref{eq: affine_xfm} for details):
\begin{align}\label{eq: lag}
\mathcal{L}(x,\lambda,\mu,{\tilde d})=&m^{T}x+\lambda^{T}(Ax+b)\nonumber\\&+\mu^{T}\left(Cx+d\right)-{\tilde d}^{T}\left(Fx+g\right).
\end{align}
Using this Lagrangian, we construct the Lagrange dual via
\begin{align}\label{eq: LagDual}
\max_{\lambda,\mu\ge0,{\tilde d}\in\mathcal{Q}_{r}}\min_{x}\;\mathcal{L}(x,\lambda,\mu,{\tilde d}).
\end{align}
We now offer two methods for solving the inner primal; the first method will be used to demonstrate the correctness of the tight duals; the second one will be proposed as a lower bounding function. The first method, which is standard for constructing a conic dual, drives the coefficient terms in front of $x$ from the Lagrangian in \eqref{eq: lag} to 0. We thus arrive at the conic dual in Model~\ref{model:all_conic_dual}, which we refer to as "All Conic Dual", since all RSOC constraints are modeled as conic constraints.
\begin{model}[h]
\caption{\hspace{-0.1cm}\textbf{:} All Conic Dual (ACD)}
\label{model:all_conic_dual}
\vspace{-0.25cm}
\begin{subequations}\label{eq: conic_soln}
\begin{align}
\max_{\lambda,\mu\ge0,{\tilde d}\in\mathcal{Q}_{r}} & \;\;\lambda^{T}b+\mu^{T}d-{\tilde d}^{T}g\\
{\rm s.t.}\quad\; & \;\;m+A^{T}\lambda+C^{T}\mu-F^{T}{\tilde d}=0.
\end{align}
\end{subequations}
\vspace{-0.25cm}
\end{model}
In the second approach, we apply the dual norm~\cite{16} to~\eqref{eq: LagDual}; before doing so, we first remove the inequality term from the Lagrangian, and we replace it with the explicit constraint $\underline{x}\le x\le\overline{x}$ under the minimizer:
\begin{align}\label{eq: lag_daul_xb}
\max_{\lambda,{\tilde d}\in\mathcal{Q}_{r}}\min_{\underline{x}\le x\le\overline{x}}\;\underbrace{\left(m^{T}+\lambda^{T}A-{\tilde d}^{T}F\right)}_{\gamma}x+\lambda^{T}b-{\tilde d}^{T}g.
\end{align}
To apply the dual norm, we normalize all variables via mean vector \( \omega = 0.5 (\overline{x} + \underline{x}) \) and range vector \( \sigma = 0.5 (\overline{x} - \underline{x}) \). 
We then define the diagonal scaling matrix \( M_\sigma = \mathrm{diag}(\sigma) \) and normalize as
\( x_n = M_\sigma^{-1}(x - \omega) \). 
This affine transformation maps all variables to the canonical $\ell_\infty$ unit ball (i.e., \( \|x_n\|_\infty \le 1 \)). Applying the substitution $x=M_{\sigma}x_n+\omega$, the new optimization formulation (see ~\cite{16}) is given by
\begin{align}\label{eq: dn_opt}
\max_{\lambda,{\tilde d}\in\mathcal{Q}_{r}}&-\left\Vert \left(m^{T}+\lambda^{T}A-{\tilde d}^{T}F\right)M_{\sigma}\right\Vert _{1}\nonumber\\&+\left(m^{T}+\lambda^{T}A-{\tilde d}^{T}F\right)\omega+\lambda^{T}b-{\tilde d}^{T}g.
\end{align}

\begin{remark}
Formulations \eqref{eq: conic_soln} and \eqref{eq: dn_opt} have the same objective values. Given a solution to \eqref{eq: dn_opt}, the full solution to \eqref{eq: conic_soln} is recoverable by solving $C^{T}\mu=F^{T}{\tilde d}-A^{T}\lambda-m$ for $\mu\ge 0$, which has a unique solution due to the structure of $C$.
\end{remark}

\section{Dual RSOC Inequality Constraints are Tight}
\label{sec:rsoctight}
Using the formulations presented in the previous section, we now transform the dual model into an equivalent expression that does not contain conic constraints; instead, these constraints are implicitly enforced in the objective function and the equality constraints. A lower bounding function will then be derived from this result.

\subsection{Explicitly enforcing dual RSOC constraints}
Both maximization problems \eqref{eq: conic_soln} and \eqref{eq: dn_opt} contain a set of decoupled dual RSOC constraints $\forall (i,j)$:
\begin{subequations}\label{eq: dual_conic_eqs}
\begin{align}
2d_{ij}^{v_{1}}d_{ij}^{v_{2}} & \ge\left(d_{ij}^{v}\right)^{T}d_{ij}^{v}:\;\text{dual Jabr voltage relaxation}\label{eq: dual_conic_eqs1}\\
2d_{ij}^{s_{1}}d_{ij}^{s_{2}} & \ge\left(d_{ij}^{s}\right)^{T}d_{ij}^{s}:\;\text{dual sending-end flow limit}\label{eq: dual_conic_eqs2}\\
2d_{ji}^{r_{1}}d_{ji}^{r_{2}} & \ge\left(d_{ji}^{r}\right)^{T}d_{ji}^{r}:\;\text{dual receiving-end flow limit}\label{eq: dual_conic_eqs3}\\
2d^{t_{1}}d^{t_{2}} & \ge\left(d^{t}\right)^{T}d^{t}:\;\;\;\,\text{dual generator cost epigraph}.\label{eq: dual_conic_eqs4}
\end{align}
\end{subequations}

\begin{lemma}
\label{lemma:1}
When \eqref{eq: conic_soln} is solved to optimality, the conic inequalities \eqref{eq: dual_conic_eqs2}-\eqref{eq: dual_conic_eqs4} are tight, i.e., they hold with equality.
\end{lemma}

\begin{proof}
We first note a structural property of the dualization map: for each RSOC
constraint of the form $2\,(\tfrac12)\,y \ge \|z\|_2^2$ appearing in
\eqref{eq:rsoc2}--\eqref{eq:rsoc4}, the row of $F$ associated with the first scalar of the dual tuple is identically zero, while the
corresponding component of $g$ equals the constant $\tfrac12$. 
Consequently, the scalars
$d_{ij}^{s_1},d_{ji}^{r_1},d^{t_1}$ do \emph{not} appear in the affine
equality $m+A^T\lambda+C^T\mu-F^T\tilde d=0$; they enter the dual
problem of \eqref{eq: conic_soln} only through the objective term $-\tilde d^T g$ and the
corresponding RSOC inequality.

We now consider the term $-{\tilde d}^{T}g$ in the objective of \eqref{eq: conic_soln}; this term collects all dual variables which encounter constant coefficients in the primal constraints of \eqref{eq: full_rsoc} when these constraints are dualized:
\begin{align}\label{eq: dtg}
    -{\tilde d}^{T}g=\sum_{ij}&\left(-\frac{1}{2}d_{ij}^{s_{1}}-S_{ij}^{{\rm max}}d_{ij}^{s_{2}}\right)\nonumber\\&+\sum_{ij}\left(-\frac{1}{2}d_{ji}^{r_{1}}-S_{ji}^{{\rm max}}d_{ji}^{r_{2}}\right)-\frac{1}{2}d^{t_{1}}.
\end{align}
Let \(s\) locally denote any scalar variable from \eqref{eq: dtg} (e.g., $d_{ij}^{s_{1}}$, $d_{ij}^{s_{2}}$, $d_{ji}^{r_{1}}$, $d_{ji}^{r_{2}}$ or $d^{t_{1}}$). In \eqref{eq: dtg}, $s$ appears with a strictly negative coefficient. Moreover, \(s\) is linked to the remaining dual variables only through the single constraint \(s \ge f(\cdot)\) appearing in \eqref{eq: dual_conic_eqs2}-\eqref{eq: dual_conic_eqs4}. Therefore, if the constraint were slack at a maximizer, we could decrease \(s\) and strictly improve the objective, contradicting optimality. Thus, each inequality in~\eqref{eq: dual_conic_eqs2}-\eqref{eq: dual_conic_eqs4} must be tight at optimality.
\end{proof}

Using this result, we may also infer the value of $d^{t_2}$ at optimality. A similar result was first reported in~\cite{16} in the DCOPF case, but here we extend this to the ACOPF case, and we offer a more complete proof.

\begin{lemma}
\label{lemma:2}
The maximizer of \eqref{eq: conic_soln} satisfies
$d^{t_2}=1$.
\end{lemma}

\begin{proof}
Assume $t$ is the first component of $x$, as defined in \eqref{eq: x_vec}. Suppose this coordinate is bounded above by a known constant $\bar t>0$ (a conservative upper bound on the maximum total quadratic cost) so that \(|t|\le \bar t\) in the primal. Now, consider the Lagrangian in \eqref{eq: lag_daul_xb}. The first component in $\gamma$ is given by
\begin{align}
\gamma_1=1-d^{t_2},
\end{align}
where $1$ is the selection element of $m$ which selects for $t$ (see \eqref{eq: m}). Since all other terms ($\gamma_2$, $\gamma_3$, etc.) correspond to different primal coordinates and therefore do not contribute to $t$. Minimizing the Lagrangian over just $t$ yields
\begin{align}
\min_{|t|\le \bar t} (1-d^{t_{2}})\,t
\;=\;
-\bar t\,|1-d^{t_{2}}|,
\end{align}
Hence, the dual objective contains the term $-\bar t\,|1-d^{t_{2}}|$. Next, by tightness of \eqref{eq: dual_conic_eqs4} at optimality,
\begin{align}
    d^{t_1}=\frac{\|d^{t}\|_2^2}{2d^{t_{2}}}, \qquad d^{t_{2}}>0
\end{align}
so the term $-0.5\,d^{t_{1}}$ in \eqref{eq: dtg} becomes
\begin{align}
-0.5\,d^{t_{1}}=-\frac{1}{4}\,\frac{\|d^t\|_2^2}{d^{t_{2}}}.
\end{align}
Holding all other dual variables fixed, the $d^{t_{2}}$-dependent part of the objective in \eqref{eq: dn_opt} (where $\omega_1=0$) is
\begin{align}
h(d^{t_{2}})=-\bar t\,|1-d^{t_{2}}|-\frac{1}{4}\frac{\|d^{t}\|_{2}^{2}}{d^{t_{2}}},\qquad d^{t_{2}}>0.
\end{align}
The function $h$ is concave and non-smooth at $d^{t_{2}}=1$. Its sub-differential is given by
\begin{align*}
\frac{\partial h}{\partial d^{t_{2}}}=\begin{cases}
{\displaystyle \bar t + \frac{1}{4}\frac{\|d^{t}\|_{2}^{2}}{\left(d^{t_{2}}\right)^{2}},} & 0<d^{t_2}<1,\\[1.2em]
{\displaystyle \Big[-\bar{t}+\frac{1}{4}\|d^{t}\|_{2}^{2},\;\bar t +\frac{1}{4}\|d^{t}\|_{2}^{2}\Big],} & d^{t_2}=1,\\[1.2em]
{\displaystyle -\bar{t}+\frac{1}{4}\frac{\|d^{t}\|_{2}^{2}}{(d^{t_2})^{2}},} & d^{t_2}>1.
\end{cases}
\end{align*}
We choose $\bar t$ sufficiently large such that$\bar t \ge \sum_{k=1}^{n_g} c_k^q (P_k^{\max})^2$, which will be large enough to ensure $\bar t \ge \tfrac14\|d_t\|_2^2$ at optimality\footnote{Since $\bar t$ is just an upper bound, it can be set arbitrarily large and have no impact on the solution of the ACOPF relaxation.}. Under this condition, every element of $\partial h$ is strictly positive for $0<d^{t_2}<1$ and strictly negative for $d^{t_2}>1$, while $0\in\partial h(1)$. Consequently, $h$ attains its maximum at $d^{t_2}=1$, and any maximizer of \eqref{eq: conic_soln} must satisfy $d^{t_2}=1$.
\end{proof}

As a final result, we also show that, at optimality, the dual of the Jabr voltage relaxation \eqref{eq: dual_conic_eqs1} is also tight.

\begin{lemma} \label{lemma:cs-tight}
Assume there exists a feasible solution to the RSOC relaxation (Model \ref{model:full_rsoc}). For each line $(i,j)$, the dual RSOC inequality of \eqref{eq: dual_conic_eqs1} is tight at the KKT solution.
\end{lemma}

\begin{proof}
Choose a line $(i,j)$ and locally abbreviate $p=(\tfrac{1}{2}w_i,w_j,[c_{ij}, s_{ij}])\in\mathcal K^{\mathrm{RSOC}}$ and $d=(d_{ij}^{v_{1}},d_{ij}^{v_{2}},d_{ij}^{v})\in\mathcal K^{\mathrm{RSOC^{*}}}$ as primal and dual variable tuples. At optimality, the KKT conditions imply $\langle p,d\rangle=\tfrac{1}{2}w_{i}d_{ij}^{v_{1}}+w_{j}d_{ij}^{v_{2}}+[c_{ij},s_{ij}]d_{ij}^{v}=0$. Recall: ($i$) $\mathcal K^{\mathrm{RSOC}}$ is a closed, pointed, self-dual cone; ($ii$) if $d\in\mathrm{int}\,\mathcal K^{\mathrm{RSOC^{*}}}$, then $\langle p,d\rangle>0$ for all $p\in\mathcal K^{\mathrm{RSOC}}\setminus\{0\}$. We now consider two mutually exclusive cases for $d$.

\textit{Case I: $d \in \mathrm{int}(\mathcal K^{\mathrm{RSOC}})$ (dual strictly inside).} 
The KKT complementarity condition $\langle p,d\rangle = 0$, together with ($ii$), forces $p = 0$. Thus $w_i=w_j=c_{ij}=s_{ij}=0$. In the Jabr mapping, this implies $v^2_i=v^2_j=0$, which contradicts the feasibility assumption. Hence, at the KKT solution, this case is impossible.

\textit{Case II: $d\in\partial\mathcal K^{\mathrm{RSOC}}$ (dual on the boundary).}
Since Case I is impossible at a feasible KKT point, $d$ must lie on the boundary. Mapping back to the dual variables for line $(i,j)$, this identity is exactly the tightness condition
\begin{align}
    2d_{ij}^{v_{1}}d_{ij}^{v_{2}}=\|d_{ij}^{v}\|_2^2.
\end{align}
If the corresponding primal RSOC constraint is loose, then $d_{ij}^{v_{1}} = d_{ij}^{v_{2}}=d_{ij}^{v}=0$, by ($ii$). In this case, the boundary identity still holds (both sides are zero).
\end{proof}

As shown in the preceding Lemmas, the dual RSOC constraints arising from the Jabr relaxations, line flows, and generator quadratic costs will always hold tightly at optimality. Since \eqref{eq: dual_conic_eqs1}-\eqref{eq: dual_conic_eqs4} hold with equality, we can enforce them explicitly by eliminating one scalar conic variable in each constraint. To accomplish this, we define the replacement function $r(z)$, which replaces a dual variable (i.e., the first dual variable in each tuple from \eqref{eq:dual_tuples}) with a nonlinear function:
\begin{equation}
\label{eq:rz}
r(z,\epsilon)\triangleq
\begin{cases}
\dfrac{\|d_{\alpha_\ell}^{\ell}\|_2^2}
      {2d_{\alpha_\ell}^{\ell_2}+\epsilon},
& z=d_{\alpha_\ell}^{\ell_1},\ \ell\in\{v,s,r\},\\[4mm]
\dfrac{\|d^t\|_2^2}{2},
& z=d^{t_1},\\[1mm]
z,
& \text{otherwise}.
\end{cases}
\end{equation}

\noindent where \(\alpha_v=\alpha_s=(i,j)\), \(\alpha_r=(j,i)\) and $\epsilon$ is a numerical stability term that will be discussed later in the paper. Applying this replacement function to the vector of dual variables from \eqref{eq: conic_soln}, we arrive at Model \ref{model:all_tight_dual}. In this new model, there are no explicit RSOC constraints; there are only non-negativity constraints. We call this the ``all tight" dual model, since all dual constraints are enforced to be tight.

\begin{model}[h]
\caption{\hspace{-0.1cm}\textbf{:} ``All Tight" Dual (ATD)}
\label{model:all_tight_dual}
\vspace{-0.25cm}
\begin{subequations}\label{eq:all_tight_dual}
\begin{align}
\max_{\lambda,\mu\ge0,{\tilde d}} & \;\;\;\;\lambda^{T}b+\mu^{T}d-r({\tilde d},\epsilon)^{T}g
\label{eq:all_tight_dual_obj}\\
{\rm s.t.}\;\; & \;\;\;\;m+A^{T}\lambda+C^{T}\mu-F^{T}r({\tilde d},\epsilon)=0
\label{eq:all_tight_dual_eq}\\
 & \;\;\;\;d_{ij}^{v_{2}},d_{ij}^{s_{2}},d_{ji}^{r_{2}}\ge0,\;\forall(i,j)\in\mathcal L.
\label{eq:all_tight_dual_nonneg}
\end{align}
\end{subequations}
\vspace{-0.25cm}
\end{model}
\subsubsection*{Equivalence of ACD and ATD}
Each nontrivial component of $r(\tilde d,\epsilon)$ takes the quadratic-over-linear form $\|d^x\|_2^2/(2 d^{x2}+\epsilon)$, which is convex on $\{d^{x2} > -\epsilon/2\}$~\cite{Boyd_Vandenberghe_2004}. Consequently, Model~\ref{model:all_tight_dual} is a smooth nonlinear program: its objective \eqref{eq:all_tight_dual_obj} is concave (a sum of linear terms and negated quadratic-over-linear terms), but the equality constraint \eqref{eq:all_tight_dual_eq} is nonconvex because the quadratic-over-linear entries enter through $F^T r(\tilde d,\epsilon)$. Strictly speaking, ATD is therefore a nonconvex NLP, whereas the ACD in Model~\ref{model:all_conic_dual} is a convex RSOC program.

This apparent loss of convexity is benign and a modeling artifact, not a loss of
global information. Lemmas~\ref{lemma:1}--\ref{lemma:cs-tight} establish that the globally optimal primal--dual pair of the \emph{convex} ACD problem already satisfies every RSOC inequality with equality, so ATD is simply a reparameterization of the same optimum with one scalar variable per cone eliminated; it does not introduce spurious stationary points near the ACD optimum, as confirmed by the numerical results in Sec.~\ref{sec:Numeric}.

The practical benefits are threefold: (i) the KKT system contains no cone-complementarity block and has fewer variables and constraints, potentially shrinking the linear system factored at each interior-point iteration; (ii) interior-point line searches no longer need to remain inside the cone interior, which is precisely where Ipopt and Knitro struggle on larger ACD instances (Table~\ref{tab:consolidated}); and (iii) the feasible set reduces to non-negativity bounds, which is the structure targeted by projected and first-order GPU methods which are the downstream motivation of this work. 

\subsection{Handling non-differentiability via $\epsilon$-stabilization and dual post-processing}
\label{sec:epsilon}

In Model~\ref{model:all_tight_dual}, if we set $\epsilon=0$, many RSOC-related terms would appear in quotient form ($i,j$ dropped for clarity)
\begin{align}
d^{{x_1}} = \frac{\|d^{x}\|_2^2}{2\,d^{{x_2}}},\qquad x\in\{v,s,r,t\}.
\label{eq:cone}
\end{align}
This quotient is well-defined only when the denominator component $d^{x_2} > 0$. To avoid numerical singularities during optimization, we have introduced a small stability parameter $\epsilon>0$. In our tests, we solve the perturbed dual problem with denominators of the form $2d^{x_2}+\epsilon$, as shown in~\eqref{eq:rz}.
\begin{figure}[h]
\centering
\includegraphics[width=\columnwidth]{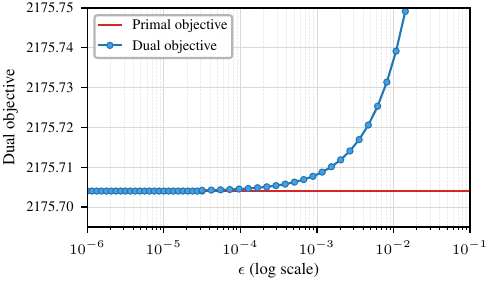}
\caption{Change of dual objective as a function of $\epsilon$. 
As the stability parameter is swept (log scale), the dual 
objective value for the 14-bus case increases monotonically.}
\label{fig:eps_sweep_case14}
\end{figure}

Unfortunately, the proposed $\epsilon$-stabilization procedure implicitly implies $2d^{x_{1}}d^{x_{2}} <\|d^{x}\|_{2}^{2}$, which is a violation of the original dual RSOC constraint.\footnote{This can be seen by taking $\epsilon>0$, setting $d^{x_{1}}=\|d^{x}\|_{2}^{2}/(2d^{x_2}+\epsilon)$, and plugging this into $2d^{x_{1}}d^{x_{2}}\ge\|d^{x}\|_{2}^{2}$: a violation of the inequality occurs.} This produces a solution which is slightly infeasible with regard to the original dual RSOC problem of Model \eqref{model:all_conic_dual}.
As a result, the $\epsilon$-perturbed dual objective is not, in general, a valid lower bound on the RSOC primal problem (i.e., weak and strong duality are both violated). This behavior is illustrated in Fig.~\ref{fig:eps_sweep_case14}, where the dual objective monotonically increases as $\epsilon$ grows. Across all tested PGLib instances, for sufficiently small $\epsilon$, the dual objective forms a reasonably flat plateau, while larger values of $\epsilon$ introduce a positive bias, resulting in increasingly loose dual objectives that significantly surpass the primal.

While using very small values of $\epsilon$ is an apparent solution to this problem, numerical solvers struggle when dual scalars $d^{x_2}$ and $\epsilon$ are both near zero. In our approach, we use a moderate value for $\epsilon$. To recover a certified lower bound for the original (unperturbed) problem, we apply a post-processing step to the solution of the $\epsilon$-stabilized dual. Denote $\bar{d}^{x_1}$, etc., as the solutions to the perturbed problem, and ${\hat d}^{x_1}$, etc., as their projected counterparts. First, tight cone feasibility is enforced by projecting the dual variables onto the appropriate RSOC faces. We consider two cases.
\begin{enumerate}
\item When a scalar component satisfies $ {\bar d}^{x_2} < \delta$, where  $\delta$ is small, the corresponding vector component is set to zero, as required by RSOC geometry:
\begin{align}
    {\hat d}^{x_{2}}={\hat d}^{x_{1}}=0, \; {\hat d}^{x}={\bm 0},\qquad x\in\{v,s,r,t\}.
\end{align}
\item When a scalar component satisfies $ {\bar d}^{x_2} \ge \delta$, then
\begin{align}
\hat{d}^{{x_1}} = \frac{\|\bar{d}^{x}\|_2^2}{2\,\bar{d}^{{x_2}}},\;\hat{d}^{x_{2}}=\bar{d}^{x_{2}},\;\hat{d}^{x}=\bar{d}^{x},\; \forall x.
\label{eq:cone_post}
\end{align}
\end{enumerate}
This yields a dual point that is feasible for the $\epsilon = 0$ problem.

After enforcing conic feasibility, we may essentially plug the resulting dual variables into the objective from \eqref{eq: dn_opt}. To state this explicitly, we collect all projected dual conic variables into $\hat{d}$. Let us now define the function $u(\lambda,\hat{d})$ via
\begin{align}
u(\lambda,\hat{d})\triangleq m^{T}+\lambda^{T}A-\hat{d}^{T}F.
\end{align}
There is no more replacement function \eqref{eq:rz} in this expression, since all duals have been projected feasible. The ACOPF relaxation is then lower bounded by
\begin{align}
f_{{\rm clb}} =& -\| u(\lambda,\hat{d})M_{\sigma}\|_{1}+u(\lambda,\hat{d})\omega+\lambda^{T}b-\hat{d}^Tg,
\end{align}
where $f_{{\rm clb}}$ is a ``certified lower bound." This procedure yields a valid dual lower bound for the original problem, where $\epsilon$ is a numerical stabilizer to facilitate dual optimization.
\section{Numerical Experiments}
\label{sec:Numeric}
To assess the performance of the proposed method, we tested the formulation on PGLib cases of 3 to 1354 buses~\cite{babaeinejadsarookolaee2019power}. We implemented a second-order cone relaxation (\texttt{SOCWRConicPowerModel}) using PowerModels.jl~\cite{PM}. All models were formulated in Julia via JuMP~\cite{JuMP}. 

The SOCWRConic relaxations were solved using MOSEK (v.~11.0.30)~\cite{MOSEK}. The SOC solutions are used as benchmarks for assessing solution accuracy.
To attain the results for the proposed models, first, we solve the \textit{primal} problem using Ipopt (v. 3.14.19)~\cite{Ipopt} and Knitro (v. 15.1.0)~\cite{Knitro}. Then, the different formulations for the \textit{dual} are tested separately with Ipopt and Knitro. 
We present Table~\ref{tab:consolidated} to show the performance of our method. The gaps are in reference to our canonical RSOC
formulation (Model~\ref{model:full_rsoc}), which is mathematically
equivalent to PowerModels' \texttt{SOCWRConicPowerModel}.
\[
{\rm Gap}=100\times
\frac{f_{\rm primal}^{\rm ref}-f_{\rm clb}}
{f_{\rm primal}^{\rm ref}},
\]
Since the All Conic Dual model includes conic constraints, Knitro is flagged to consider it as a conic problem. For transparency, the status labels in the tables are reported as returned
by the solvers. 
Table~\ref{tab:consolidated} shows the certified lower bounds and their gaps with the primal SOCP. For ATD, these gaps are reported after dual feasibility restoration.

These numerical experiments are intended primarily to validate the
structural equivalence between ATD and ACD established in
Sec.~\ref{sec:rsoctight}. The reported gaps are measured against the
corresponding primal reference objective, and therefore
serve as an accuracy check rather than a direct claim of CPU runtime
superiority. The SOCWRConic model solved with MOSEK is included as a
mature primal conic benchmark to provide a CPU solve-time baseline.

As reported in Table~\ref{tab:consolidated}, ATD consistently recovers
the primal reference objective to solver tolerance when the solver
converges, while being substantially easier to solve than ACD on the
larger instances. For example, in the 500-bus case, Ipopt solves ATD
to a 0.00\% gap in 5.70\,s, whereas ACD reaches the iteration limit
with a 98.79\% gap; Knitro solves ATD with a 0.00\% gap in 4.38\,s,
compared with 61.21\,s for ACD. The 1354-bus case shows a similar
pattern: Knitro solves ATD to a 0.00\% gap in 15.64\,s, while ACD
requires 703.78\,s and attains a 0.03\% gap. With Ipopt, both
formulations reach the iteration limit, but ATD returns a substantially
smaller gap than ACD, 18.00\% versus 52.39\%.
These results show that ATD preserves the primal benchmark
objective while replacing the ACD cone structure with a formulation
that is numerically easier for nonlinear solvers and more amenable to
future GPU-parallel first-order implementations.
\setlength{\textfloatsep}{4pt}
\setlength{\floatsep}{4pt}
\setlength{\intextsep}{4pt}
\begin{table}
\caption{Solver performance across PGLib cases. Gap (\%) is
measured against the primal value. Markers for solver status when non-convergence happens:
$^\dagger$\textsc{iter\_lim}, $^\ddagger$\textsc{infeas},
$^\S$\textsc{slow}.}
\label{tab:consolidated}
\centering
\footnotesize
\setlength{\tabcolsep}{2.6pt}
\renewcommand{\arraystretch}{1.05}
\begin{tabular*}{\columnwidth}{@{\extracolsep{\fill}} l l c c c c c c}
\toprule
& & \multicolumn{3}{c}{\textbf{All Tight Dual}} & \multicolumn{2}{c}{All Conic Dual} 
& \multicolumn{1}{c}{\makecell[c]{SOCWRConic\\(MOSEK)}} \\
\cmidrule(lr){3-5}\cmidrule(lr){6-7}
Case & Solver & Time (s) & Gap (\%) & \makecell{Post.\\(ms)}
& Time (s) & Gap (\%) & \multicolumn{1}{c}{Time (s)} \\
\midrule
\midrule
\multirow{2}{*}{3}
& Ipopt  & 0.07              & 0.00 & 0.01 & 0.09              & 0.00  & \multirow{2}{*}{0.01} \\
& Knitro & \textbf{0.00 }             & 0.00 & 0.02 & 0.00              & 0.02  &       \\
\midrule
\multirow{2}{*}{14}
& Ipopt  & 0.07              & 0.00 & 0.02 & 0.50$^\ddagger$   & 0.68  & \multirow{2}{*}{0.01} \\
& Knitro & \textbf{0.03}              & 0.00 & 0.02 & 0.09              & 0.00  &       \\
\midrule
\multirow{2}{*}{57}
& Ipopt  & 0.29              & 0.00 & 0.04 & 4.43$^\ddagger$   & 0.35  & \multirow{2}{*}{0.05} \\
& Knitro & \textbf{0.20}              & 0.00 & 0.07 & 1.27              & 0.00  &       \\
\midrule
\multirow{2}{*}{118}
& Ipopt  & 1.35              & 0.00 & 0.06 & 1.03$^\dagger$    & 11.93 & \multirow{2}{*}{0.14} \\
& Knitro & \textbf{1.27}              & 0.00 & 0.06 & 7.77$^\S$         & 0.00  &       \\
\midrule
\multirow{2}{*}{300}
& Ipopt  & 58.48             & 0.69 & 0.14 & 129.30$^\dagger$  & 32.19 & \multirow{2}{*}{0.50} \\
& Knitro & {\textbf{2.13}}     & {0.00} & 0.14 & 240.46 & 0.36 &       \\
\midrule
\multirow{2}{*}{500}
& Ipopt  & {5.70}     & {0.00} & 0.25 & 121.16$^\dagger$  & 98.79 & \multirow{2}{*}{0.94} \\
& Knitro & {\textbf{4.38}}     & {0.00} & 0.26 & 61.21             & 0.00  &       \\
\midrule
\multirow{2}{*}{1354}
& Ipopt  & 4985.35$^\dagger$ & 18.00 & 0.70 & 533.76$^\ddagger$  & 52.39 & \multirow{2}{*}{3.43$^\S$} \\
& Knitro & {\textbf{15.64}}    & {0.00} & 0.81 & 703.78            & 0.03  &       \\
\bottomrule
\end{tabular*}
\vspace{1mm}

\begin{minipage}{0.95\columnwidth}
\scriptsize
\emph{Note:} The post-processing column reports the time (ms) required for certified lower-bound recovery.
\end{minipage}
\end{table}



\section{Conclusion}
\label{sec:conclusion}

This work developed the All Tight Dual (ATD) model for the Jabr-relaxed ACOPF problem by exploiting RSOC tightness to remove auxiliary conic variables and simplify the dual feasible set. Across PGLib benchmarks spanning small to large transmission networks, ATD produces high-quality dual bounds with minimal optimality gap relative to established conic relaxations, while exhibiting improved numerical stability compared with the All Conic Dual (ACD), especially on larger instances where interior-point methods can struggle. To ensure correctness despite $\epsilon$-stabilization and solver tolerances, we also propose a one-time post-processing step that projects the dual variables onto the appropriate RSOC faces and provides a certified lower bound, which will be valid independent of the choice of $\epsilon$. An open-source Julia implementation of this certified lower bound is available in OPFCertify.jl~\cite{opfcertify}.


This work does not aim to benchmark CPU wall-clock performance against mature commercial conic solvers such as MOSEK. Instead, it focuses on developing and analyzing an accurate, simplified, and certified dual lower-bounding function tailored for first-order methods, with the goal of enabling efficient GPU acceleration and ultimately achieving competitive runtimes relative to state-of-the-art conic solvers. Therefore, these results motivate, as a next step, the design of first-order algorithms tailored to ATD that are amenable to GPU acceleration for scalable ACOPF.
\appendix
\subsection{Affine function for RSOC Dualization}\label{eq: affine_xfm}
As a guiding example, we may dualize the RSOC constraint $2\left(\tfrac{1}{2}\right)S_{ji}^{{2,{\rm max}}}\ge P_{ji}^{2}+Q_{ji}^{2}$ using the dual variable tuple $(d_{ji}^{r_{1}},d_{ji}^{r_{2}},d_{ji}^{r})$, which is collected into the dual variable vector $\tilde d$. We now define the affine transformation
\begin{align}
    \underbrace{\left[\begin{array}{cc}
0 & 0\\
0 & 0\\
1 & 0\\
0 & 1
\end{array}\right]}_{F}\underbrace{\left[\begin{array}{c}
P_{ji}\\
Q_{ji}
\end{array}\right]}_{x}+\underbrace{\left(\begin{array}{c}
\tfrac{1}{2}\\
S_{ji}^{{\rm max}}\\
0\\
0
\end{array}\right)}_{g}.
\end{align}
Thus, the dualization inner product is then given as
\begin{align}
\tilde{d}^{T}(Fx+g)=\tfrac{1}{2}d_{ji}^{r_{1}}+S_{ji}^{{\rm max}}d_{ji}^{r_{2}}+[P_{ji},Q_{ji}]d_{ji}^{r}.
\end{align}
We apply this procedure to the full set of primal and dual variables for all RSOC constraints to yield the expression $\tilde{d}^{T}(Fx+g)$ in \eqref{eq: lag}.


\makeatletter
\renewcommand{\IEEEbibitemsep}{0pt plus .2pt}
\makeatother

\let\oldthebibliography\thebibliography
\let\endoldthebibliography\endthebibliography
\renewenvironment{thebibliography}[1]{%
  \footnotesize
  \oldthebibliography{#1}
  \setlength{\itemsep}{0pt}
  \setlength{\parskip}{0pt}
}{\endoldthebibliography}

\bibliographystyle{IEEEtran}
\bibliography{references}
\end{document}